\documentclass[nofootinbib,prl,aps,reprint,amsmath,amssymb]{revtex4-1}
\usepackage{hyperref}
\usepackage{changes}
\usepackage[section]{placeins}
\usepackage[titletoc]{appendix}
\usepackage{xcolor}
\usepackage{dsfont}
\usepackage{bm}
\usepackage{mathtools}
\usepackage{enumerate}
\DeclareMathAlphabet{\mathpzc}{OT1}{pzc}{m}{it}

\newcommand{\sayy}[1]{`#1'}
\providecommand{\href}[2]{#2}

\usepackage{amsthm}
\newtheorem{theorem}{Theorem}

\def\be{\begin{equation}}
\def\ee{\end{equation}}
\def\bea{\begin{eqnarray}}
\def\eea{\end{eqnarray}}

\def\sig{\sigma}
\def\hsig{\hat{\sigma}}

\def\la{\langle}
\def\ra{\rangle}
\def\Eu{ \mathfrak{H} }

\def\obs{\mathcal{O}}

\definecolor{MyB}{rgb}{0.1,0.1,1.0}

\def\jcap{{J. Cosmol. Astropart. Phys.} }

\def\apjl{{Astrophys. J. Lett.}}

\def\MNRAS{{Mon.\ Not.\ R.\ Astr.\ Soc.}}

\def\aap{{Astron.\ Astrophys}}

\begin{document}
\title{Reconciling a decelerating Universe with cosmological observations} 
\author{Asta~Heinesen}
\email{asta.heinesen@ens--lyon.fr}
\affiliation{Univ Lyon, Ens de Lyon, Univ Lyon1, CNRS, Centre de Recherche Astrophysique de Lyon UMR5574, F--69007, Lyon, France}

\begin{abstract}
Can modern cosmological observations be reconciled with a general-relativistic Universe without an anti-gravitating energy source?  
Usually, the answer to this question by cosmologists is in the negative, and it is commonly believed that the observed excess dimming of supernovae relative to that in the Milne model is evidence for dark energy. 
In this paper, we develop {theorems that clarify} the conditions for such an excess dimming, based on which we argue that {the} answer to the above question may counter-intuitively be \sayy{yes}. 

\end{abstract}

\maketitle

\section{Introduction}
\label{sec:intro} 
Dark energy  
first became an established part of the cosmological paradigm in 1998, when data from supernovae of type Ia showed a greater dimming in their luminosity with redshift than what could be explained with a Milne universe model\footnote{The Milne model (empty FLRW model without a cosmological constant) defines an upper limit for the distance-redshift curve of the expanding general-relativistic Friedmann-Lema\^itre-Robertson-Walker (FLRW) models that obey the strong energy condition. }  
\cite{SupernovaSearchTeam:1998fmf,SupernovaCosmologyProject:1998vns}. When interpreted within the general-relativistic FLRW models, the data thus indicated the existence of dark energy, an anti-gravitating energy source violating the universally
attractive nature of gravity known from ordinary matter. 
Shortly after this realisation, it was proposed that the dimming of light from supernovae could alternatively be explained in universe models with ordinary matter forming a large-scale cosmic inhomogeneity as modelled by the Lema\^itre-Tolman-Bondi (LTB) models \cite{Pascual-Sanchez:1999xpt,Celerier:1999hp}. 
These models exhibit universal \emph{deceleration} of distances between geodesic test particles, but an observer who is placed properly relative to the inhomogeneity can infer the same {trend in} dimming of light from supernovae as is observed in an FLRW model with an accelerated scale factor; see \cite{2010A&A...518A..21C} for a review.  
The LTB metrics that can account for the observed supernova luminosities are challenged by complementary data \cite{Bull:2011wi,Redlich:2014gga} and are furthermore breaking with the Copernican principle. 
It is of fundamental interest if models that satisfy the strong energy condition \emph{and} the Copernican principle could produce an excess dimming of light from supernovae relative to the Milne model. 
{Despite of the potentially far reaching implications for fundamental physics if there is indeed a competitive cosmological model of this type, there has not yet been any systematic classifications made of space-time geometries that exhibit super-Milne dimming of light. } 

In this paper, we formulate {two theorems that bound the} angular diameter distance (luminosity distance) without employing conjectures for light propagation or any \emph{a priori} constraints on the underlying geometry; we only assume the geometrical optics description for light and general relativity as the gravitational theory. 
{The theorems define geometrical conditions (necessary and sufficient, respectively) for the angular diameter distance to exceed that of a Milne universe model, thus providing the first general classification of models that qualify for describing the observations of supernova lightcurves and complementary probes of cosmic distances.} 
{Based on these theorems, we discuss the circumstances under which a super-Milne dimming of light measured by Copernican observers may occur without dark energy.}


\section{Observed distances and redshifts in a general space-time}
\label{sec:zdist} 

\subsection{Assumptions and definitions} 
Consider a space-time manifold with a metric $g_{\mu \nu}$ of signature $(- + + +)$ and a Levi-Civita connection $\nabla_\mu$. 
Let a {beam} of light pass from a source to an observer in this space-time, and let the event of observation be labelled $\obs$. 
We assume that the geometrical optics approximation holds, such that the photons of the {beam} can be described as {null geodesic} test particles with 4-momentum $k^\mu$. 
We consider a \sayy{cosmic congruence} field with 4-velocity $u^\mu$ that is defined in the space-time neighbourhood of the {null beam, and make the unique decomposition}    
\bea
\label{kdecomp}
k^\mu = E (u^\mu - e^\mu) \, , \qquad E \equiv - u^\mu k_\mu \,  , 
\eea 
where $e^\mu$ is a spatial unit vector that is orthogonal to the cosmic 4-velocity field $u^\mu e_\nu = 0$. 
{We define the cosmological redshift and angular diameter distance as 
\bea
z  \equiv  \frac{E}{E_\obs} -1 \, , \qquad d_A \equiv \sqrt{\delta A}/\sqrt{\delta\Omega} \, , 
\eea 
respectively, 
where $\delta A$ is the physical area of the emitting object perpendicular to the direction of emission of the null ray and $\delta\Omega$ is its angular size relative to the observer at $\obs$ comoving with $u^\mu$.}  
We may assume that the photon number in the light bundle is preserved on the path from the emitter to the observer, such that Etherington's reciprocity theorem  \cite{2007GReGr..39.1047E} holds for deriving luminosity distance:  $d_L = (1+z)^2  d_A$.  In this case, it suffices to analyse angular diameter distance and redshift, from which the value of luminosity distance follows. 

\subsection{Evolution of cosmic redshift} 
The evolution of the cosmic redshift, $z$, along the null ray is given by 
\bea
\label{def:zprime}
\hspace*{-0.2cm} \frac{ {\rm d} z}{{\rm d} \lambda} = - E_\obs (1+z)^2  \Eu \, ,   
\eea 
where $\frac{ {\rm d} }{{\rm d} \lambda} \! \equiv \! k^\nu \nabla_\nu$ is the directional derivative along the null ray, and where we have introduced the \sayy{effective Hubble parameter} 
\bea
\label{def:Eu}
\hspace*{-0.2cm}  \Eu \equiv      \frac{ {\rm d} E^{-1}}{{\rm d} \lambda}    =   \frac{1}{3}\theta  - e^\mu a_\mu + e^\mu e^\nu \sigma_{\mu \nu} \, ,
\eea 
which reduces to the Hubble parameter \sayy{$\dot{a}/a$} in an FLRW geometry. We have made use of the kinematic decomposition 
\bea
\label{def:expu} 
&& \hspace{-0.3cm} \nabla_{\nu}u_\mu  = \frac{1}{3}\theta h_{\mu \nu }+\sig_{\mu \nu} + \omega_{\mu \nu} - u_\nu a_\mu \ , \nonumber \\ 
&& \hspace{-0.3cm}  \theta \equiv \nabla_{\mu}u^{\mu} ,  \; \;   \sig_{\mu \nu} \equiv  h_{ \la  \nu  }^{\, \beta}  h_{  \mu \ra }^{\, \alpha } \nabla_{  \beta }u_{\alpha  }  ,  \; \;    \omega_{\mu \nu} \equiv h_{ [  \nu  }^{\, \beta}  h_{  \mu ] }^{\, \alpha } \nabla_{  \beta}u_{\alpha  }     , 
\eea 
where $h_{\mu \nu} = g_{\mu \nu} + u_{\mu} u_{\nu}$ is the spatial projection tensor orthogonal to $u^\mu$, and 
where the triangular bracket $\la \ra$ around indices selects the tracefree symmetric part of the spatial tensor and $[ ]$ selects the anti-symmetric part. 
The variables $\theta$, $\sig_{\mu \nu}$, and $\omega_{\mu \nu}$ describe respectively the volume expansion, shear, and vorticity of the cosmic congruence, and $a^\mu \equiv  \dot{u}^\mu$ is the 4-acceleration of the individual observers in the congruence, where the overdot $\dot{} \! \equiv \! u^\nu \nabla_\nu$ represents the directional derivative along the cosmic flow-lines. 
The second derivative of $z$ is: 
\bea
\label{def:zdoubleprime}
\frac{ {\rm d^2} z}{{\rm d} \lambda^2} =  E_\obs^2   \left( 3 + \mathfrak{Q} \right) \Eu^2  (1+z)^3  \, ,  
\eea 
where the \sayy{effective deceleration parameter} 
\bea 
\label{def:dec}
 \mathfrak{Q}  &\equiv&    - 1 - \frac{1}{E} \frac{     \frac{ {\rm d} \Eu}{{\rm d} \lambda}    }{\Eu^2}  \, , 
\eea 
reduces to the usual FLRW deceleration parameter \sayy{$- a \ddot{a}/\dot{a}^2  $} in an FLRW geometry. 
In general, $\mathfrak{Q}$ can be expressed as \cite{Umeh:2013UCT,Heinesen:2020bej} 
\bea
\label{q}
\hspace*{-0.5cm}  \mathfrak{Q} &=&  - 1 -  \frac{ \overset{0}{\mathfrak{q}}   +  \bm{e} \cdot  \bm{{\overset{1}{\mathfrak{q}}}}   +    \bm{e} \bm{e} \cdot  \bm{{\overset{2}{\mathfrak{q}}}}     +    \bm{e} \bm{e} \bm{e} \cdot  \bm{{\overset{3}{\mathfrak{q}}}}    +    \bm{e} \bm{e} \bm{e} \bm{e} \cdot  \bm{{\overset{4}{\mathfrak{q}}}}   }{\Eu^2}    \, , 
\eea 
where $\bm{e} \cdot  \bm{{\overset{1}{\mathfrak{q}}}} \equiv e^\mu \overset{1}{\mathfrak{q}}_\mu$, $\bm{e} \bm{e} \cdot  \bm{{\overset{2}{\mathfrak{q}}}} \equiv e^\mu e^\nu \overset{2}{\mathfrak{q}}_{\mu \nu}$, etc., with multipole coefficients 
\bea
\label{qpoles}
&& \overset{0}{\mathfrak{q}} \equiv  \frac{1}{3}   \dot{\theta}  + \frac{1}{3} D_{   \mu} a^{\mu  } - \frac{2}{3}a^{\mu} a_{\mu}    - \frac{2}{5} \sigma_{\mu \nu} \sigma^{\mu \nu}    \, , \nonumber \\ 
&& \overset{1}{\mathfrak{q}}_\mu \equiv   -   \dot{a}_\mu  - \frac{1}{3} D_{\mu} \theta  + a^\nu \omega_{\mu \nu}  +  \frac{9}{5}  a^\nu \sigma_{\mu \nu}    -  \frac{2}{5}   D_{  \nu} \sigma^{\nu }_{\;  \mu  }   \, , \nonumber \\
&& \overset{2}{\mathfrak{q}}_{\mu \nu}  \equiv     \dot{\sigma}_{\mu \nu}    +   D_{  \la \mu} a_{\nu \ra } + a_{\la \mu}a_{\nu \ra }     - 2 \sigma_{\alpha (  \mu} \omega^\alpha_{\; \nu )}   - \frac{6}{7} \sigma_{\alpha \la \mu} \sigma^\alpha_{\; \nu \ra }   \, , \nonumber \\ 
 && \overset{3}{\mathfrak{q}}_{\mu \nu \rho}  \equiv -  D_{ \la \mu} \sigma_{\nu   \rho \ra }    -  3  a_{ \la \mu} \sigma_{\nu \rho \ra }    \, , \quad       \overset{4}{\mathfrak{q}}_{\mu \nu \rho \kappa}  \equiv  2   \sigma_{\la \mu \nu } \sigma_{\rho \kappa \ra} \, , 
\eea 
where $D_\mu$ is the covariant spatial derivative as projected onto the 3-dimensional space orthogonal to $u^\mu$. 
Apart from the volume acceleration term, $\propto \dot{\theta}$, there are a number of terms arising from anisotropic and inhomogeneous universe kinematics, which are generally not constrained in sign or amplitude by general-relativistic energy conditions. 
The interpretation of $\mathfrak{Q}$ as a direct measure of the deceleration of distances between test particles is thus generally \emph{not} valid except for in the strict FLRW case.  
While the dimensionless effective deceleration parameter $\mathfrak{Q}$ may become singular in regions where $\Eu = 0$, the \emph{dimensionful} effective deceleration parameter $\Eu^2 \mathfrak{Q}$ remains finite when the first and second derivatives of $z$ are finite.

 \subsection{Evolution of cosmic distance}
The evolution of the cosmic angular diameter distance, $d_A$, along a null ray of the congruence is given by the focusing equation, cf. equation~44 in~\cite{Perlick:2010zh}, 
\be
\label{focuseq}
\frac{{\rm d}^2 d_A }{ {\rm d} \lambda^2} 
= - {\mathcal{F}} d_{A}  \, , \quad  \mathcal{F} \equiv \left(\frac{1}{2} \hsig^{\mu \nu} \hsig_{\mu \nu} + \frac{1}{2}k^{\mu}k^\nu R_{\mu \nu}  \right)
\ee 
where $\hsig_{\mu \nu}$ is the shear tensor of the photon congruence, and $R_{\mu \nu}$ is the Ricci curvature of the space-time{, and where we have introduced $\mathcal{F}$ as a short hand notation for the combined focusing.}  
The evolution of $d_A$ may be solved for with knowledge of shear and Ricci curvature by use of the initial conditions at the vertex of the observer's lightcone 
\be
\label{dAinit}
d_A \rvert_\obs = 0 \, , \qquad \frac{{\rm d} d_A }{ {\rm d} \lambda } \rvert_\obs = - E_\obs \, , 
\ee  
{which can formally be obtained by expanding the Jacobi map around the observer  \cite{Seitz:1994xf}.} 
We have from \eqref{def:zprime} that 
\be
\label{dAinitz}
 \frac{{\rm d} z }{ {\rm d} d_A  } \Big\rvert_\obs = \Eu_\obs \, , 
 \ee
which gives the observational interpretation of $\Eu_\obs$ as the slope of the redshift--distance function at the observer.

\section{Distance concavity in general-relativistic space-times} 

\subsection{Bounds on the dimming of light} 

Let us consider general-relativistic space-time scenarios where the null energy condition is satisfied, meaning $k^\mu k^\nu R_{\mu \nu} \geq 0$. 
{From this it immediately follows from \eqref{focuseq} that there is a positive focus of the bundle, and 
t}ogether with the initial condition \eqref{dAinit} this means that 
\be
\label{dAprimebound}
\frac{{\rm d} d_A }{ {\rm d} \lambda } \geq - E_\obs  \,  
\ee 
at any point along the null ray, where we recall that $\lambda$ is increasing towards the observer. 
{We shall prove the following theorems that bound the observed angular diameter distance (or luminosity distance) as a function of the observed redshift.}

\begin{theorem}[\textbf{Sub-Milne dimming}] 
\label{ths:second}
Consider a general-relativistic space-time obeying the null energy condition. 
For a null geodesic congruence with $\Eu_\obs > 0$, the following applies: 
The angular diameter distance, $d_A$, is bounded from above by the Milne universe model in terms of its redshift for a section of the null geodesic path $[\lambda_1,\lambda_2]$ if 
\be
\label{Qident}
\hspace{-0.05cm}  \la \! \la \mathfrak{Q}  \Eu^2 \ra \! \ra   \geq  0   \, ,  \qquad   \la \! \la \mathfrak{Q}  \Eu^2 \ra \! \ra \equiv   \frac{\int_{\lambda}^{\lambda_\obs} \!  {\rm d} \lambda'  \int_{\lambda'}^{\lambda_\obs} \!  {\rm d} \lambda''  \mathfrak{Q}  \Eu^2 }{\frac{1}{2} (\lambda_\obs - \lambda )^2} \, , 
\ee  
is satisfied for all $\lambda \in [\lambda_1,\lambda_2]$, corresponding to $d_A \in [d_A(\lambda_2), d_A(\lambda_1)]$. 
\end{theorem}

\begin{proof} 
It follows from \eqref{dAprimebound} that the angular diameter distance is bounded in terms of the affine distance along the null geodesic 
\be
\label{dAlambda}
d_A \leq  E_\obs (\lambda_\obs - \lambda) \,  . 
\ee 
We use \eqref{def:zprime} to make the rewriting 
\be
\label{eq:milnez}
\hspace{-0.25cm} 
\frac{ 1 - \frac{1}{(1+z)^2}}{2}  =   - \!\! \int^{\lambda_\obs}_\lambda  \!\!\!\! {\rm d} \lambda'  \frac{\frac{ {\rm d} z}{{\rm d} \lambda'} }{(1+z)^3} = E_\obs \!\! \int^{\lambda_\obs}_\lambda \!\!\!\! {\rm d} \lambda'  \frac{\Eu}{1+z} \, . 
\ee 
Together with the identity $\frac{\Eu}{1+z} - \Eu_\obs = E_\obs  \int^{\lambda_\obs}_\lambda \!\!  {\rm d} \lambda'   \Eu^2 \, \mathfrak{Q}$ (which follows from multiplying \eqref{def:dec} with $\Eu^2$ and integrating both sides) this yields 
\be
\label{dAexact}
\hspace{-0.19cm}  \frac{ 1 \! - \! \frac{1}{(1+z)^2}}{2 }  \! =\!  \Eu_\obs E_\obs (\lambda_\obs \! -\!  \lambda) \! \! \left(\! 1  +   \frac{\! E_\obs (\lambda_\obs \!  - \! \lambda) \la \! \la \mathfrak{Q}  \Eu^2 \ra \! \ra \! }{2 \Eu_\obs}       \! \right) \!  . 
\ee 
Combining \eqref{dAlambda} with \eqref{dAexact}, and assuming that \eqref{Qident} holds for a distance interval $d_A \in [d_A(\lambda_2), d_A(\lambda_1)]$, this gives the following bound for distances in that interval:  
\be
\label{dAlambdaMilne2}
\Eu_\obs  d_A  \leq  \frac{ 1}{2} \! \left( \! 1 - \frac{1}{(1+z)^2} \! \right)
\ee 
where the right hand side is the dimensionless angular diameter distance in the Milne model, cf., e.g., equation~6~in \cite{Chodorowski:2005wj}. 

\end{proof} 

\noindent
\textbf{Remarks to Theorem~\ref{ths:second}.} 
\noindent 
{Modern cosmological observations violate the bound \eqref{dAlambdaMilne2} for small and intermediate values of redshift. 
Assuming that general relativity is the correct gravitational theory, that the null energy condition is satisfied, and that the geometrical optics approximation holds for describing light, Theorem~\ref{ths:second} thus implies that the condition \eqref{Qident} must be violated at these scales, which again implies that the dimensionful deceleration parameter, $\Eu^2 \mathfrak{Q}$, must be mostly negative along the null rays.  
In FLRW cosmology, the effective deceleration parameter reduces to the scale factor deceleration such that $\Eu^2 \mathfrak{Q} = $ \sayy{$- \ddot{a}/a$}, which, if negative, immediately implies   existence of a gravitationally repulsive source within general relativistic FLRW solutions. 
However, the effective deceleration parameter has a non-trivial interpretation for general space-time geometries, and it does \emph{not} follow from this theorem that the strong energy condition must be violated to explain the observed excess dimming. 
In the below we formulate a sufficient condition for super-Milne dimming. }

{
\begin{theorem}[\textbf{Super-Milne dimming}] 
\label{ths:super} 
Consider as in Theorem~\ref{ths:second} a general-relativistic space-time obeying the null energy condition. 
For a null geodesic congruence with $\Eu_\obs > 0$, the following applies: 
The angular diameter distance, $d_A$, is bounded from below by the Milne universe model in terms of its redshift for a section of the null geodesic path $[\lambda_1,\lambda_2]$, if 
\be
\label{Qavdef}
\hspace{-0.2cm} \frac{   \la \! \la \mathfrak{Q}  \Eu^2 \ra \! \ra }{\Eu_\obs^2}    \leq - 2 \frac{1}{\Eu_\obs \mathcal{L}}    \, , 
\ee 
is satisfied for all $\lambda \in [\lambda_1,\lambda_2]$, corresponding to $d_A \in [d_A(\lambda_2), d_A(\lambda_1)]$, 
where 
\be
\label{L}
\mathcal{L}   \equiv   \frac{1}{ d_A }   \frac{  E_\obs^2 }{ \la \mathcal{F} \ra  }  \, ,  \qquad \la \mathcal{F} \ra  \equiv \frac{\int_{\lambda}^{\lambda_\obs} \!  {\rm d} \lambda' \mathcal{F} }{\lambda_\obs - \lambda}  \, , 
\ee  
is a length scale set by the accumulated focusing of the null beam and the distance to the emitter. 

\end{theorem} 

\begin{proof} 
We write the angular diameter distance as 
\be
\label{dAremainder}
d_A = E_\obs (\lambda_\obs - \lambda) + \Delta d_A   \, , 
\ee 
where $\Delta d_A$ is the remainder term of the first order series expansion of $d_A$, which by \eqref{focuseq} can be expressed as 
\bea
\label{remainder}
&&\hspace{-0.585cm} \Delta d_A  =\! \int_{\lambda}^{\lambda_\obs} \!\!\!\!\!\! {\rm d} \lambda' (\lambda' \! - \! \lambda) \frac{{\rm d}^2 d_A }{ {\rm d} \lambda'^2} =  - \!\! \int_{\lambda}^{\lambda_\obs} \!\!\!\!\!\! {\rm d} \lambda' (\lambda' \! - \! \lambda)  \mathcal{F} d_{A} ,    
\eea 
The null energy condition dictates that $d_A(\lambda')\leq d_A(\lambda)$ for $\lambda' \in [\lambda, \lambda_\obs]$. Exploring this, and that $\lambda' \leq \lambda_\obs$ for the same interval, we have   
\bea
\label{remainder2}
 \Delta d_A   
 && \, \, \geq \, -  E^2_\obs (\lambda_\obs - \lambda)^2  \mathcal{L}^{-1}  \, , 
\eea
which when inserted in \eqref{dAremainder} gives 
\be
\label{dAremainder2}
d_A \,  \geq \, E_\obs (\lambda_\obs - \lambda) \left(1 -   E_\obs (\lambda_\obs - \lambda)  \mathcal{L}^{-1}   \right)  \, . 
\ee 
We now use \eqref{dAexact} in combination with the result in \eqref{dAremainder2} to obtain 
\be
\label{dAlower}
\hspace{-0.1cm} \frac{ 1 - \frac{1}{(1+z)^2}}{2}  \leq \Eu_\obs d_A \frac{1  +    \frac{E_\obs (\lambda_\obs  -  \lambda) \la \! \la \mathfrak{Q}  \Eu^2 \ra \! \ra }{2 \Eu_\obs}  }{1 -  \frac{E_\obs (\lambda_\obs - \lambda)}{  \mathcal{L}}     }    \, \,  . 
\ee 
Assuming that \eqref{Qavdef} holds for an interval $d_A \in [d_A(\lambda_2), d_A(\lambda_1)]$, this gives the following bound for distances in that interval:  
\be
\label{dAlambdaMilnesuper}
\Eu_\obs  d_A  \geq  \frac{ 1}{2} \! \left( \! 1 - \frac{1}{(1+z)^2} \! \right) . 
\ee 

\end{proof} 

\noindent
\textbf{Interpretation of the length scale $\mathcal{L}$.} 
\noindent 
The inverse length scale $\mathcal{L}^{-1}$ is zero when the light propagates in empty space (such that $k^\mu k^\nu R_{\mu\nu} = 0$) and in the absence of Weyl focusing of the beam (such that $\hat{\sigma}_{\mu\nu} = 0$). 
From an order of magnitude estimate in a dust dominated universe with negligible Weyl focusing along the beam, we have $\mathcal{F}/E^2 \approx 4\pi G \rho$, where $\rho$ is the mass density of the dust and $G$ is the gravitational constant. 
We thus have from the definition \eqref{L} that $\mathcal{L}^{-1} \sim d_A  4\pi G \la  \rho /(1+z)^2 \ra$. 
For typical light beams in expanding universe models, the average $4\pi G \la  \rho /(1+z)^2 \ra$ is usually smaller than or of order the square of the expansion rate $\theta_\obs$, which in turn sets a characteristic inverse length scale for cosmology (Hubble length scale in FLRW cosmology). 
When the emitter is in the cosmic vicinity of the observer such that $d_A \ll \theta^{-1}_\obs$ we have $\mathcal{L}^{-1} \lesssim d_A \theta^2_\obs \ll \theta_\obs$, and in such cases we thus expect  
the right hand side of \eqref{Qavdef} to be small, except for the case of extremely lensed light beams.  

\vspace{5pt}
\noindent
\textbf{Remarks to Theorem~\ref{ths:super}.} 
\noindent 
A trivial way to satisfy \eqref{Qavdef} is to introduce dark energy or another energy-momentum source that violates the strong energy condition, but here we shall focus on situations where the strong energy condition holds. 
To satisfy \eqref{Qavdef} and the strong energy condition simultaneously, there must be non-trivial local violations of the FLRW idealisation with imprints on the redshift of light that do not cancel along typical null beams.

Systematic departures of $\Eu^2 \mathfrak{Q}$ from the length scale deceleration $(\theta/3)^2(-1 \! - \!  3  \dot{\theta}  / \theta^2)$ can occur if $e^\mu$ aligns systematically with the multipole coefficients in \eqref{qpoles}.       
This happens for instance in the non-Copernican LTB void models\footnote{Concretely, the terms $\propto e^\mu D_\mu \theta$ (see dipole coefficient in \eqref{qpoles}) give systematic negative contributions in $\mathfrak{Q}$ when the photon is propagating towards the center of the void (towards less density and faster volume expansion).} \cite{Pascual-Sanchez:1999xpt,Celerier:1999hp}, which is the underlying reason why these models can produce a breaking of the bound \eqref{dAlambdaMilne2} for a central observer. 
In general, the spatial propagation direction $e^\mu$ of the photons is determined by the equation 
\bea
\label{kderive}
\hspace{-0.6cm} \frac{h^{\mu}_{\, \nu}  k^\alpha  \nabla_\alpha e^\nu }{E}  \!   = \!  e^\mu e^\nu e^\rho \sigma_{\nu \rho} \! - \! e^\nu \sigma^{ \mu}_{\; \nu}  \! - \! e^\nu  \omega^{\mu}_{\; \nu}  \!-\! e^\mu e^\nu a_\nu \! + \! a^\mu    , 
\eea  
which can be obtained from the geodesic equation $k^\nu \nabla_\nu k^\mu \! =\! 0$ and the decomposition~\eqref{kdecomp}. 
The differential equation~\eqref{kderive} makes explicit how the photon direction of propagation is responding to the cosmic kinematics.  
Concretely, $e^\mu$ is driven towards alignment with $a^\mu$ and eigendirections of $\sigma^{\mu}_{\; \nu}$, while $\omega^{\mu}_{\; \nu}$ contributes with a deflection effect perpendicular to the initial direction of propagation of the photon bundle.  
In addition, shear causes squeezing of structures along the shear eigendirections, and photons that propagate along axes of more expansion (positive shear) will tend to spend more time in the structure than photons that propagate along axes of less expansion (negative shear).  
These systematic trends caused by anisotropic expansion  
hold the potential to produce non-cancelling effects that may satisfy the excess-dimming condition \eqref{Qavdef} in a universe that is everywhere locally decelerating.

}

\subsection{Discussion of model scenarios} 
In linearly-perturbed FLRW models, there is systematic alignment of $e^\mu$ with shear eigendirections, but the resulting systematic contributions from the term $e^\mu e^\nu \sigma_{\mu \nu}$ cancels  with systematic contributions from perturbations in the expansion rate $\theta$ in \eqref{def:Eu}, such that redshift is well modelled by the FLRW background expansion rate in these models  \cite{Rasanen:2011bm}. 
Similar cancelations have been noted for Swiss-cheese models \cite{Koksbang:2017arw,Rasanen:2011bm}. 
Consequently, the redshift measured by observers comoving with the matter in Swiss-cheese models with LTB and Szekeres structures tends to be well modelled by the FLRW background model \cite{Fleury:2014gha,Szybka:2010ky,Koksbang:2021zyi}, although counter examples exist \cite{Lavinto:2013exa}. 

The accurate predictions of redshift by the FLRW {relation in models where this would not \emph{a priori} have been expected, may be understood through a convenient choice of cosmic reference frame in these models.}   
If there is \emph{a} cosmic reference frame that is kinematically close to a reference FLRW model {such that $\theta/3$ is sufficiently close-to-homogeneous and such that $\sigma_{\mu \nu}/\theta$ and $a^\mu/\theta$ have small norms, and which is in addition close to the observers and emitters of light in terms of a relative Lorentz boost},  
it can be shown from \eqref{def:zprime} and \eqref{def:Eu} that the redshift is close to that {predicted by the FLRW reference model}. 

Hypersurface forming frames where shear is almost vanishing and expansion of space is almost homogeneous, thus inheriting the properties of the Poisson gauge from linearised FLRW perturbation theory, {have been shown to exist in a variety of models, including} certain LTB solutions\footnote{This is not in contradiction with the fact that the LTB models can account for supernova dimming for a central observer in a large underdensity \cite{2009JCAP...09..022E}.} 
\cite{VanAcoleyen:2008cy}, post-Newtonian perturbation theory \cite{Clifton:2020oqx}, and non-linear numerical simulation studies \cite{Giblin:2018ndw}.  
This may be the reason why various (analytical and simulated) models {that have large local density contrasts relative to an FLRW reference model are well described by the FLRW distance--redshift relation (or more generally the Dyer-Roeder relation\footnote{The Dyer-Roeder approximation \cite{1972ApJ...174L.115D,Dyer:1973zz} asserts that the average redshift of the light is given by the background FLRW model, and that the only systematic contribution to the angular diameter distance from lumpiness of structure along the light beam comes from a scaling of the mean density of matter.} in cases of unfair sampling of the density of matter by the light)} \cite{Biswas:2007gi,Fleury:2014gha,Giblin:2016mjp,Sanghai:2017yyn,East:2017qmk,Adamek:2018rru,Macpherson:2022eve}.   

We remark that a space-time model that conforms to the Dyer-Roeder approximation with an expanding FLRW background reference frame (i.e., $\Eu > 0$ for the background) subject to the strong energy condition (i.e., $\mathfrak{Q} \geq 0$ for the background) obeys the {sub-}Milne bound \eqref{dAlambdaMilne2} for the average observed $d_A$ and $z$ of the model. 
It is therefore imperative to examine Copernican space-time scenarios in which the Dyer-Roeder approximation is not accurate\footnote{See \cite{1986A&A...168...57E,PhysRevD.40.2502,Rasanen:2008be,Clarkson:2011br,Fleury:2014gha,Koksbang:2020qry,Koksbang:2021zyi} for investigations into the validity of the Dyer-Roeder approximation in various model scenarios.}, if the observed dimming of supernovae is to be explained without dark energy.  
Such model scenarios have been examined \cite{Lavinto:2013exa,Sikora:2016rmq,Koksbang:2021zyi}, and while these should not be considered competitive cosmological models, they do present proof-of-concept models where the distance--redshift relation is distorted away from the Dyer-Roeder prediction. 
{Such models are of great interest to explore further, with particular focus on models that satisfy the condition \eqref{Qavdef}.}

\section{Conclusion} 
\label{sec:discussion} 
We have investigated light propagation in the geometrical optics limit of general relativity, and examined, from first principles, the conditions under which cosmological observations can be made compatible with a Universe with ordinary matter and radiation only. 
The effective deceleration parameter $\mathfrak{Q}$, which determines the observed dimming of supernovae, was first discussed in detail in \cite{Clarkson:2011br} where the implications for the interpretation of supernova data were also addressed. 
{In this paper we have formulated Theorem~\ref{ths:second} and Theorem~\ref{ths:super}, which present a necessary and a sufficient condition, respectively, for the observed super-Milne dimming {of light from} supernovae.}    
Systematic effects from inhomogeneities have previously been proposed to hold the potential to mimic dark energy through backreaction effects on global volume dynamics \cite{Buchert:2018vqi}, such as the backreaction functionals proposed by~Buchert in \cite{Buchert:1999er,Buchert:2001sa,Buchert:2019mvq}. 
An upper bound on the expansion rate of the cosmic fluid frame in terms of the Milne expansion rate was formulated by R\"as\"anen for general-relativistic irrotational-dust space-times~\cite{Rasanen:2005zy}. This bound however, does not imply that the distance--redshift relation in such space-times is bounded from above by the Milne model -- a counter example can be constructed with an LTB model, as discussed in the introduction.  

In this paper, we have analysed observables directly; concretely the stated theorems apply to measurements that probe angular diameter distance (or luminosity distance) and redshift. 
In order to arrive at a systematic excess dimming effect relative to that of a uniform universe model, photons must systematically \sayy{pick up} local irregularities of the space-time as they propagate through it.  
As discussed in the above, we \emph{do} generally expect photons to systematically align with eigendirections of the kinematic variables of the cosmic reference frame, and such alignments hold the potential to cause systematic contributions in the distance--redshift relation that mimic dark energy in our Universe. 
{Alignments of this type are indeed present in the $\Lambda$CDM model, but the resulting leading order correction terms to the background FLRW redshift function cancel, cf. \cite{Rasanen:2011bm}, 
and most modern cosmological simulations agree well with the FLRW prediction of redshift.   
Exceptions of models exist, however, where Copernican observers measure redshifts with systematic departures from those predicted by the background FLRW model. While it remains to show that a competitive cosmological model that exhibits the discussed excess dimming can be formulated, it is a possibility that we live in a Universe that exhibits such properties. }   
In the motivation of this paper, we have focused on the dimming of light from supernovae, but the derivations apply to any observational probe of redshift and angular diameter (or luminosity) distance.

\vspace{6pt} 
\begin{acknowledgments}
This work is part of a project that has received funding from the European Research Council (ERC) under the European Union's Horizon 2020 research and innovation programme (grant agreement ERC advanced grant 740021--ARTHUS, PI: Thomas Buchert). 
I would like to thank Thomas Buchert, Syksy~R\"{a}s\"{a}nen and David~L.~Wiltshire for useful comments. 
\end{acknowledgments}

\bibliographystyle{mnras}
\bibliography{refs}

\end{document}